\documentclass[oneside,reqno,american]{amsart}
\usepackage[T1]{fontenc}
\usepackage[utf8]{inputenc}
\usepackage{xcolor}
\usepackage{babel}
\usepackage{prettyref}
\usepackage{amstext}
\usepackage{amsthm}
\usepackage{amssymb}
\usepackage[pdfusetitle,
 bookmarks=true,bookmarksnumbered=false,bookmarksopen=false,
 breaklinks=false,pdfborder={0 0 1},backref=false,colorlinks=false]
 {hyperref}
\hypersetup{
 colorlinks=true,citecolor=blue,linkcolor=blue,linktocpage=true}

\makeatletter
\numberwithin{equation}{section}
\numberwithin{figure}{section}

\usepackage{prettyref}

\newrefformat{cor}{Corollary~\ref{#1}}
\newrefformat{subsec}{Section~\ref{#1}}
\newrefformat{lem}{Lemma~\ref{#1}}
\newrefformat{thm}{Theorem~\ref{#1}}
\newrefformat{sec}{Section~\ref{#1}}
\newrefformat{chap}{Chapter~\ref{#1}}
\newrefformat{prop}{Proposition~\ref{#1}}
\newrefformat{exa}{Example~\ref{#1}}
\newrefformat{tab}{Table~\ref{#1}}
\newrefformat{rem}{Remark~\ref{#1}}
\newrefformat{def}{Definition~\ref{#1}}
\newrefformat{fig}{Figure~\ref{#1}}
\newrefformat{claim}{Claim~\ref{#1}}

\makeatother

\theoremstyle{plain}
\newtheorem{thm}{\protect\theoremname}[section]
\theoremstyle{remark}
\newtheorem{rem}[thm]{\protect\remarkname}
\theoremstyle{plain}
\newtheorem{cor}[thm]{\protect\corollaryname}
\theoremstyle{definition}
\newtheorem{example}[thm]{\protect\examplename}
\theoremstyle{remark}
\newtheorem*{conclusion*}{\protect\conclusionname}
\newrefformat{cor}{Corollary \ref{#1}}
\newrefformat{enu}{Item~\ref{#1}}
\newrefformat{exa}{Example \ref{#1}}
\newrefformat{rem}{Remark \ref{#1}}
\newrefformat{sec}{Section~\ref{#1}}
\newrefformat{subsec}{Section~\ref{#1}}
\newrefformat{thm}{Theorem \ref{#1}}
\providecommand{\conclusionname}{Conclusion}
\providecommand{\corollaryname}{Corollary}
\providecommand{\examplename}{Example}
\providecommand{\remarkname}{Remark}
\providecommand{\theoremname}{Theorem}

\begin{document}
\title[graph based Operator Inequalities]{Graph Structured Operator Inequalities and Tsirelson-Type Bounds}
\begin{abstract}
We establish operator norm bounds for bipartite tensor sums of self-adjoint
contractions. The inequalities generalize the analytic structure underlying
the Tsirelson and CHSH bounds, giving dimension-free estimates expressed
through commutator and anticommutator norms. A graph based formulation
captures sparse interaction patterns via constants depending only
on graph connectivity. The results link analytic operator inequalities
with quantum information settings such as Bell correlations and network
nonlocality, offering closed-form estimates that complement semidefinite
and numerical methods.
\end{abstract}

\author{James Tian}
\address{Mathematical Reviews, 535 W. William St, Suite 210, Ann Arbor, MI
48103, USA}
\email{james.ftian@gmail.com}

\maketitle
\tableofcontents{}

\section{Introduction}

We study universal operator norm bounds for bipartite tensor sums
\[
B=\sum^{m}_{i=1}x_{i}\otimes y_{i},
\]
where each $x_{i}$ and $y_{i}$ is a self-adjoint contraction. The
starting point is the simple identity underlying the CHSH and Tsirelson
bounds: for self-adjoint unitaries $A_{0},A_{1},B_{0},B_{1}$ with
$\left[A_{i},B_{j}\right]=0$ for all $i,j$, 
\begin{equation}
\mathcal{B}:=A_{0}B_{0}+A_{0}B_{1}+A_{1}B_{0}-A_{1}B_{1}
\end{equation}
satisfies 
\begin{equation}
\mathcal{B}^{2}=4I-\left[A_{0},A_{1}\right]\left[B_{0},B_{1}\right],\label{eq:1-2}
\end{equation}
which yields the Tsirelson bound $2\sqrt{2}$ once commutator norms
are estimated \cite{PhysRevLett.23.880}. This identity, first noted
by Cirel'son (Tsirelson), remains central in the modern analysis of
Bell inequalities and quantum correlations \cite{MR577178,MR798740,MR1907485,MR2515619,PhysRevA.83.032123}.
Indeed, taking norms in \prettyref{eq:1-2} gives $\left\Vert \mathcal{B}\right\Vert ^{2}=\left\Vert \mathcal{B}^{2}\right\Vert \leq4+\left\Vert \left[A_{0},A_{1}\right]\right\Vert \left\Vert \left[B_{0},B_{1}\right]\right\Vert \leq4+2\cdot2=8$,
so $\left\Vert \mathcal{B}\right\Vert \leq2\sqrt{2}$.

The same commutator/anticommutator expansion has intrinsic operator-theoretic
value. Sharp inequalities for commutators and anticommutators in unitarily
invariant norms have been studied extensively, from classical matrix
analysis to recent refinements. Examples include bounds by Bhatia
and Kittaneh, the Böttcher-Wenzel inequality and its variants, and
work on commutator estimates for normal or positive operators \cite{MR1607862,MR2446625}.
The proofs developed here draw directly on this analytic approach
but remain fully dimension-free.

Operator space and noncommutative probability techniques also connect
norm control of mixed products to the structure of quantum correlations.
Results such as the operator-space Grothendieck theorem, noncommutative
Khintchine inequalities, and analyses of XOR games show how tensor
norms and commutators govern achievable correlation strengths \cite{MR1930886,MR1150376,MR2006539}.
These perspectives link operator inequalities with the geometry of
quantum information theory.

In quantum information applications, Bell inequalities now extend
far beyond the two-setting CHSH case. Device independent protocols,
semidefinite-programming hierarchies, and network or many-body generalizations
provide powerful but often computationally heavy tools. The inequalities
presented here offer a complementary, closed-form approach: graph-sensitive,
dimension-free estimates that directly quantify the role of commutators
and anticommutators. For a fixed choice of coefficients and observables,
the same Bell operators $B_{c}$ appear as objects of interest in
semidefinite-programming relaxations, device-independent quantum key
distribution, and graph-based Bell tests; see, for instance, \cite{PhysRevLett.98.010401,RevModPhys.86.419,PhysRevA.78.062112,MR3036982}.

A second motivation arises from structured interactions. In spin system
and quantum network models, operators such as $\sigma_{x}\otimes\sigma_{x}+\sigma_{y}\otimes\sigma_{y}+\sigma_{z}\otimes\sigma_{z}$
represent Heisenberg type couplings, while general sums $\sum_{i}x_{i}\otimes y_{i}$
appear in lattice or graph coupled systems. The graph based inequalities
developed here provide explicit norm estimates that scale with the
interaction pattern via a simple combinatorial parameter. Concretely,
we establish complete graph bounds of the form 
\[
\left\Vert B\right\Vert ^{2}\leq m+\frac{1}{2}\sum_{i<j}\left(\left\Vert \left[x_{i},x_{j}\right]\right\Vert \left\Vert \left[y_{i},y_{j}\right]\right\Vert +\left\Vert \left\{ x_{i},x_{j}\right\} \right\Vert \left\Vert \left\{ y_{i},y_{j}\right\} \right\Vert \right),
\]
with equality for anticommuting Clifford families. The argument extends
directly to weighted sums
\[
\left\Vert B_{c}\right\Vert ^{2}\leq\sum_{i}c^{2}_{i}+\sum_{i<j}\left|c_{i}c_{j}\right|\phi_{ij},
\]
and to sparse graphs satisfying an edge domination condition that
controls non-edge interactions by averages over neighboring edges.

The graph based inequalities form the second major theme of this work.
They show that the operator norm of $B$ can be bounded by the “local”
commutation structure encoded in a graph $G=\left(V,E\right)$, with
a constant $C\left(G\right)=\frac{2\left(m-1\right)}{\delta}-1$ depending
only on the minimum degree $\delta$. This gives a quantitative link
between operator-norm growth and graph connectivity: dense graphs
recover the complete graph constant $C\left(G\right)=1$, while sparse
graphs yield controlled relaxations. The same reasoning extends to
the weighted setting, where each term carries a scalar amplitude $c_{i}$,
producing a unified framework that interpolates between universal
and graph local inequalities.

These graph based results show how combinatorial sparsity constrains
noncommutative norm growth. They align with recent developments in
graph-theoretic quantum correlations, network Bell inequalities, and
the study of commuting graph structures in operator algebras \cite{MR3036982,MR4018517,MR1450565}.
They also parallel the role of graph Laplacians and adjacency operators
in noncommutative harmonic analysis and matrix-valued inequalities
\cite{MR1477662,MR3837109}.

The resulting inequalities yield two types of consequences. First,
they provide explicit, analytic upper bounds on bipartite correlators
$\sup_{\rho}\left|tr\left(\rho B_{c}\right)\right|$ in terms of noncommutativity,
matching the Tsirelson value in the Clifford case. Second, they yield
graph dependent bounds that show how the degree structure limits or
amplifies collective correlations, allowing one to infer the presence
of many substantially noncommuting pairs from an observed Bell value.
This directly connects to current work on graph based nonlocality
and self-testing schemes using anticommuting observables \cite{PhysRevLett.117.070402,MR2090174}.

Our results thus complement existing semidefinite and numerical frameworks
by providing simple, verifiable operator inequalities that (i) scale
transparently with the number of settings, (ii) make commutator and
anticommutator dependence explicit, and (iii) adapt naturally to weighted
or sparse architectures. They also offer quick analytical estimates
for tensor-sum operators in quantum information and operator theory,
where numerical optimization may be unnecessary or infeasible. For
background we refer to standard texts in quantum information, operator
theory, and operator spaces, see e.g., \cite{MR1976867,Watrous_2018}.
In addition, graph-theoretic uncertainty relations and Lovasz-type
bounds for families of (anti)commuting observables have recently been
obtained in \cite{MR4613815,MR4751492}, where sums of squared expectations
$\sum_{i}\langle A_{i}\rangle^{2}$ are controlled via an (anti)commutativity
graph and its Lovasz theta number or semidefinite refinements. The
present work is complementary: instead of bounding $\sum_{i}\langle A_{i}\rangle^{2}$,
we give explicit commutator/anticommutator-based bounds on the operator
norm of tensor sums $\sum_{i}c_{i}x_{i}\otimes y_{i}$, including
weighted and graph structured variants for arbitrary self-adjoint
contractions. 

\section{Preliminaries}\label{sec:prelim}

We collect the basic graph and operator-theoretic notation used throughout
the paper. The results in the subsequent sections are stated for families
of self-adjoint contractions indexed by $\{1,\dots,m\}$, and the
graph structured bounds are formulated in terms of a simple graph
on this index set. For the convenience of readers from either operator
theory or graph theory, we state the relevant conventions explicitly.

\subsection{Graphs on the index set}

\label{subsec:prelim-graphs}

Fix an integer $m\geq1$ and set 
\[
V=\{1,\dots,m\}.
\]
A (simple, undirected) graph on $V$ is a pair $G=(V,E)$ where the
edge set $E\subseteq\{\{i,j\}:1\leq i<j\leq m\}$ contains no loops
and no multiple edges. For $i\in V$, the neighbor set is 
\[
N(i)=\{j\in V:\{i,j\}\in E\},
\]
the degree is $\deg(i)=|N(i)|$, and the minimum degree is 
\[
\delta=\min_{i\in V}\deg(i).
\]

We encode interaction patterns by such graphs on the index set. Here
vertex $i\in V$ labels the $i$-th summand $x_{i}\otimes y_{i}$
in the tensor sum $B=\sum^{m}_{i=1}x_{i}\otimes y_{i}$. Given $G=(V,E)$,
an edge $\{i,j\}\in E$ means that the pair $(i,j)$ is treated as
a directly coupled interaction in our estimates, with its contribution
measured by the quantity $\phi_{ij}$ defined below. Non-edge pairs
$\{i,j\}\notin E$ are not discarded, but are controlled indirectly
via the edge-domination condition in \prettyref{thm:c1}. The sparsity
or density of $G$ is quantified by the minimum degree $\delta$,
which governs the constant $C(G)$ in our graph-dependent bounds.

\subsection{Operator notation}

\label{subsec:prelim-ops}

Let $H$ and $K$ be complex Hilbert spaces. We write $B(H)$ for
the bounded operators on $H$ and use $\left\Vert \cdot\right\Vert $
for the operator norm. An operator $x\in B(H)$ is a self-adjoint
contraction if $x=x^{*}$ and $\left\Vert x\right\Vert \leq1$. For
$x,y\in B(H)$ we use the commutator and anticommutator conventions
\[
[x,y]=xy-yx,\qquad\{x,y\}=xy+yx.
\]
For self-adjoint operators $X,Y\in B(H)$, we write $X\leq Y$ in
the Löwner order to mean that $Y-X$ is positive semidefinite.

We view $B(H)\otimes B(K)$ as acting on the Hilbert space tensor
product $H\otimes K$ in the usual way, and we write $x\otimes y$
for the elementary tensor. In the graph-dependent estimates we use
the mixed interaction quantities: for $i\neq j$ and self-adjoint
contractions $x_{i},x_{j}\in B(H)$, $y_{i},y_{j}\in B(K)$, set 
\[
\phi_{ij}=\frac{1}{2}\left(\left\Vert \left[x_{i},x_{j}\right]\right\Vert \left\Vert \left[y_{i},y_{j}\right]\right\Vert +\left\Vert \left\{ x_{i},x_{j}\right\} \right\Vert \left\Vert \left\{ y_{i},y_{j}\right\} \right\Vert \right)\geq0.
\]
These $\phi_{ij}$ measure the size of the mixed terms in the expansion
of $B^{2}$ and provide the basic edge weights in the complete and
sparse graph bounds. 
\begin{rem}
A key feature of the sparse graph formulation is that the graph $G=(V,E)$
is not merely decorative: it encodes which pairs of settings are treated
as directly interacting in the operator bound. The vertex set $V=\{1,\dots,m\}$
indexes the summands $x_{i}\otimes y_{i}$, while the edge set $E$
selects a subset of index pairs whose mixed terms $\phi_{ij}$ are
retained explicitly in the estimate. Pairs $\{i,j\}\notin E$ are
not ignored, but are required to be controlled in terms of neighboring
edge pairs via the edge domination condition. In this way the choice
of $E$ specifies the “interaction skeleton’’ for the inequality,
and the graph constant $C(G)$ quantifies how much non-edge interaction
can be absorbed into edge contributions as a function of the minimum
degree $\delta$.

This point of view allows one to pass from universal inequalities
to graph-adapted ones. For a complete graph $K_{m}$ every pair is
an edge, and \prettyref{thm:b1} recovers a fully symmetric bound
in which all $\phi_{ij}$ appear. For a chain, star, or lattice, the
natural choice of $E$ is dictated by the geometry: nearest neighbours
in a one-dimensional chain, hub--leaf pairs in a star, or nearest
neighbors in a lattice or network. The sparse graph theorem \prettyref{thm:c1}
then shows that, provided non-edge interactions are dominated in the
prescribed sense, one can control $\left\Vert B\right\Vert $ using
only the $\phi_{ij}$ on this chosen edge set, at the cost of an explicit
factor $C(G)$ depending on the connectivity of $G$. This clarifies
how the combinatorial choice of $E$ influences the strength of the
operator inequality, and why the graph structure is intrinsic to the
second family of results.
\end{rem}

\section{Complete and Sparse Graph Operator Bounds}\label{sec:2}

This section develops the main operator norm inequalities for tensor
sums of self-adjoint contractions. We begin with the complete graph
case, where every pair of indices interacts, leading to a universal
bound that depends only on the pairwise commutator and anticommutator
norms. This gives a global inequality valid for all finite families
of self-adjoint contractions.

We then derive a sparse graph extension, in which the interactions
are restricted to the edges of a fixed graph $G=(V,E)$ on $\{1,\dots,m\}$.
Here $V$ indexes the summands $x_{i}\otimes y_{i}$ and $E$ specifies
which pairs $(i,j)$ are treated as edges whose interaction terms
$\phi_{ij}$ are retained explicitly in the bound. Under a natural
domination condition on the off-edge terms, the same operator norm
formulation yields a controlled estimate whose scaling depends on
the minimum degree of the graph. 
\begin{thm}
\label{thm:b1}Let $x_{1},\dots,x_{m}\in B\left(H\right)$ and $y_{1},\dots,y_{m}\in B\left(K\right)$
be self-adjoint contractions. Set 
\[
B=\sum^{m}_{i=1}x_{i}\otimes y_{i}.
\]
Then 
\begin{equation}
\left\Vert B\right\Vert ^{2}\leq m+\frac{1}{2}\sum_{i<j}\left(\left\Vert \left[x_{i},x_{j}\right]\right\Vert \left\Vert \left[y_{i},y_{j}\right]\right\Vert +\left\Vert \left\{ x_{i},x_{j}\right\} \right\Vert \left\Vert \left\{ y_{i},y_{j}\right\} \right\Vert \right).\label{eq:B1}
\end{equation}
In particular, if for each pair $\left(i,j\right)$ either $\left\{ x_{i},x_{j}\right\} =0$
or $\left\{ y_{i},y_{j}\right\} =0$, then 
\begin{equation}
\left\Vert B\right\Vert ^{2}\leq m+\frac{1}{2}\sum_{i<j}\left\Vert \left[x_{i},x_{j}\right]\right\Vert \left\Vert \left[y_{i},y_{j}\right]\right\Vert .\label{eq:B2}
\end{equation}
\end{thm}

\begin{proof}
Expand 
\[
B^{2}=\sum^{m}_{i=1}x^{2}_{i}\otimes y^{2}_{i}+\sum_{i<j}\left(x_{i}x_{j}\otimes y_{i}y_{j}+x_{j}x_{i}\otimes y_{j}y_{i}\right).
\]
Define 
\[
D:=\sum^{m}_{i=1}x^{2}_{i}\otimes y^{2}_{i},\qquad T_{ij}:=x_{i}x_{j}\otimes y_{i}y_{j}+x_{j}x_{i}\otimes y_{j}y_{i}\quad\left(i<j\right).
\]
Then $B^{2}=D+\sum_{i<j}T_{ij}$. 

Each $x_{i},y_{i}$ is a self-adjoint contraction, so $0\leq x^{2}_{i}\leq I$
and $0\leq y^{2}_{i}\leq I$ in the Löwner order. Hence $0\leq x^{2}_{i}\otimes y^{2}_{i}\leq I$
and therefore $D\leq mI.$

Next, write 
\[
T_{ij}=\frac{1}{2}\left(\left\{ x_{i},x_{j}\right\} \otimes\left\{ y_{i},y_{j}\right\} +\left[x_{i},x_{j}\right]\otimes\left[y_{i},y_{j}\right]\right).
\]
The anticommutators $\left\{ x_{i},x_{j}\right\} $ and $\left\{ y_{i},y_{j}\right\} $
are self-adjoint, so $\left\{ x_{i},x_{j}\right\} \otimes\left\{ y_{i},y_{j}\right\} $
is self-adjoint. The commutators $\left[x_{i},x_{j}\right]$ and $\left[y_{i},y_{j}\right]$
are skew-adjoint, and hence $\left[x_{i},x_{j}\right]\otimes\left[y_{i},y_{j}\right]$
is also self-adjoint. Thus $T_{ij}$ is self-adjoint. 

For any self-adjoint operator $A$, we have $-\left\Vert A\right\Vert I\leq A\leq\left\Vert A\right\Vert I$.
Using this, we get 
\begin{align*}
\left\{ x_{i},x_{j}\right\} \otimes\left\{ y_{i},y_{j}\right\}  & \leq\left\Vert \left\{ x_{i},x_{j}\right\} \right\Vert \left\Vert \left\{ y_{i},y_{j}\right\} \right\Vert I,\\
\left[x_{i},x_{j}\right]\otimes\left[y_{i},y_{j}\right] & \leq\left\Vert \left[x_{i},x_{j}\right]\right\Vert \left\Vert \left[y_{i},y_{j}\right]\right\Vert I.
\end{align*}
Therefore, 
\[
T_{ij}\leq\frac{1}{2}\left(\left\Vert \left\{ x_{i},x_{j}\right\} \right\Vert \left\Vert \left\{ y_{i},y_{j}\right\} \right\Vert +\left\Vert \left[x_{i},x_{j}\right]\right\Vert \left\Vert \left[y_{i},y_{j}\right]\right\Vert \right)I.
\]

Finally, 
\begin{align*}
B^{2} & =D+\sum_{i<j}T_{ij}\\
 & \leq mI+\frac{1}{2}\sum_{i<j}\left(\left\Vert \left\{ x_{i},x_{j}\right\} \right\Vert \left\Vert \left\{ y_{i},y_{j}\right\} \right\Vert +\left\Vert \left[x_{i},x_{j}\right]\right\Vert \left\Vert \left[y_{i},y_{j}\right]\right\Vert \right)I.
\end{align*}
Taking norms and using $\left\Vert B\right\Vert ^{2}=\left\Vert B^{2}\right\Vert $
(since $B^{*}=B$) gives the claimed inequality \eqref{eq:B1}. The
assertion \eqref{eq:B2} follows by omitting the anticommutator terms.
\end{proof}

The complete graph bound in \prettyref{thm:b1} treats the fully coupled
case, in which every pair contributes to the mixed-term expansion
of $B^{2}$. In many structured settings, however, only a subset of
pairs interact. For example, when the operators $x_{i}$ and $y_{i}$
are coupled according to a network or sparsity pattern. To capture
such partial coupling, we introduce a graph based formulation. The
next result shows that, under an edge domination hypothesis relating
non-edge to edge interactions, the same operator inequality method
extends to sparse graphs with an explicit combinatorial factor.
\begin{thm}
\label{thm:c1}Let $x_{1},\dots,x_{m}\in B\left(H\right)$ and $y_{1},\dots,y_{m}\in B\left(K\right)$
be self-adjoint contractions. Define 
\[
B=\sum^{m}_{i=1}x_{i}\otimes y_{i}.
\]
For $i\neq j$, set 
\[
\phi_{ij}=\frac{1}{2}\left(\left\Vert \left[x_{i},x_{j}\right]\right\Vert \left\Vert \left[y_{i},y_{j}\right]\right\Vert +\left\Vert \left\{ x_{i},x_{j}\right\} \right\Vert \left\Vert \left\{ y_{i},y_{j}\right\} \right\Vert \right)\geq0.
\]
Let $G$ be a simple undirected graph on $\left\{ 1,\dots,m\right\} $
with edge set $E\left(G\right)$, neighbor set $N\left(i\right)$,
degree $\deg\left(i\right)=\left|N\left(i\right)\right|$, and minimum
degree $\delta:=\min_{i}\deg\left(i\right)\geq1$. 

Assume the following edge domination condition holds for every $\left(i,j\right)\notin E\left(G\right)$:
\begin{equation}
\phi_{ij}\leq\frac{1}{\deg\left(i\right)}\sum_{k\in N\left(i\right)}\phi_{ik}+\frac{1}{\deg\left(j\right)}\sum_{k\in N\left(j\right)}\phi_{jk}.\label{eq:B3}
\end{equation}
Then 
\begin{equation}
\left\Vert B\right\Vert ^{2}\leq m+C\left(G\right)\sum_{\left(i,j\right)\in E\left(G\right)}\phi_{ij},\label{eq:B4}
\end{equation}
where 
\begin{equation}
C\left(G\right)=\frac{2\left(m-1\right)}{\delta}-1.\label{eq:bb3}
\end{equation}
\end{thm}

\begin{proof}
We follow the proof of \prettyref{thm:b1} by writing 
\[
B^{2}=D+\sum_{i<j}T_{ij},
\]
where $D=\sum^{m}_{i=1}x^{2}_{i}\otimes y^{2}_{i}\leq mI$, and $T_{ij}=\frac{1}{2}\left(\left\{ x_{i},x_{j}\right\} \otimes\left\{ y_{i},y_{j}\right\} +\left[x_{i},x_{j}\right]\otimes\left[y_{i},y_{j}\right]\right)$.
Then 
\[
\left\Vert T_{ij}\right\Vert \leq\frac{1}{2}\left(\left\Vert \left[x_{i},x_{j}\right]\right\Vert \left\Vert \left[y_{i},y_{j}\right]\right\Vert +\left\Vert \left\{ x_{i},x_{j}\right\} \right\Vert \left\Vert \left\{ y_{i},y_{j}\right\} \right\Vert \right)=\phi_{ij}.
\]
Since $T^{*}_{ij}=T_{ij}$, we get 
\[
T_{ij}\leq\left\Vert T_{ij}\right\Vert I\leq\phi_{ij}I.
\]
Summing over $i<j$ gives 
\[
\sum_{i<j}T_{ij}\leq\left(\sum_{i<j}\phi_{ij}\right)I.
\]
It follows that 
\[
B^{2}=D+\sum_{i<j}T_{ij}\leq mI+\left(\sum_{i<j}\phi_{ij}\right)I=\left(m+\sum_{i<j}\phi_{ij}\right)I.
\]
Taking operator norms (note $B^{*}=B$, so $\left\Vert B\right\Vert ^{2}=\left\Vert B^{2}\right\Vert $),
we get 
\begin{equation}
\left\Vert B\right\Vert ^{2}\leq m+\sum_{i<j}\phi_{ij}.\label{eq:bb5}
\end{equation}

Next, split the sum into edges and non-edges: 
\begin{equation}
\sum_{i<j}\phi_{ij}=\sum_{i<j,\:\left(i,j\right)\in E\left(G\right)}\phi_{ij}+\underset{=:S_{NE}}{\underbrace{\sum_{i<j\:\left(i,j\right)\not\in E\left(G\right)}\phi_{ij}}}.\label{eq:BB5}
\end{equation}

Applying the assumption \eqref{eq:B3} to each non-edge pair of indices
gives 
\[
S_{NE}\leq\sum_{i<j,\:\left(i,j\right)\notin E\left(G\right)}\left(\frac{1}{\deg\left(i\right)}\sum_{k\in N\left(i\right)}\phi_{ik}+\frac{1}{\deg\left(j\right)}\sum_{k\in N\left(j\right)}\phi_{jk}\right).
\]
Writing $F\left(i\right)=\frac{1}{\deg\left(i\right)}\sum_{k\in N\left(i\right)}\phi_{ik}$,
and summing over unordered non-edges, we have
\begin{align*}
S_{NE} & \leq\sum_{i<j,\:\left(i,j\right)\notin E\left(G\right)}\left(F\left(i\right)+F\left(j\right)\right)=\frac{1}{2}\sum_{i\neq j,\:\left(i,j\right)\notin E\left(G\right)}\left(F\left(i\right)+F\left(j\right)\right)\\
 & =\sum_{i\neq j,\:\left(i,j\right)\notin E\left(G\right)}F\left(i\right)=\sum_{i\neq j,\:\left(i,j\right)\notin E\left(G\right)}\frac{1}{\deg\left(i\right)}\sum_{k\in N\left(i\right)}\phi_{ik},
\end{align*}
since the ordered sum has each unordered pair twice, and the two terms
are symmetric. 

For fixed $i$, the number of non-neighbors of $i$ is $m-1-\deg\left(i\right)$,
thus 
\[
\sum_{i\neq j,\:\left(i,j\right)\notin E\left(G\right)}\frac{1}{\deg\left(i\right)}\sum_{k\in N\left(i\right)}\phi_{ik}=\sum^{m}_{i=1}\frac{m-1-\deg\left(i\right)}{\deg\left(i\right)}\sum_{k\in N\left(i\right)}\phi_{ik}
\]
and so 
\begin{equation}
S_{NE}\leq\sum^{m}_{i=1}\frac{m-1-\deg\left(i\right)}{\deg\left(i\right)}\sum_{k\in N\left(i\right)}\phi_{ik}.\label{eq:B5}
\end{equation}
Since $\deg\left(i\right)\geq\delta$ by assumption, one has 
\[
\frac{m-1-\deg\left(i\right)}{\deg\left(i\right)}=\frac{m-1}{\deg\left(i\right)}-1\leq\frac{m-1}{\delta}-1.
\]
Substitute this into \eqref{eq:B5}, 
\begin{equation}
S_{NE}\leq\left(\frac{m-1}{\delta}-1\right)\sum^{m}_{i=1}\sum_{k\in N\left(i\right)}\phi_{ik}=2\left(\frac{m-1}{\delta}-1\right)\sum_{\left(i,j\right)\in E\left(G\right)}\phi_{ij}.\label{eq:B6}
\end{equation}
Insert \eqref{eq:B6} into \eqref{eq:BB5}, we get 
\[
\sum_{i<j}\phi_{ij}\leq\left(\frac{2\left(m-1\right)}{\delta}-1\right)\sum_{\left(i,j\right)\in E\left(G\right)}\phi_{ij}.
\]

Finally, returning to \eqref{eq:bb5}, we have 
\[
\left\Vert B\right\Vert ^{2}\leq m+\sum_{i<j}\phi_{ij}\leq m+\left(\frac{2\left(m-1\right)}{\delta}-1\right)\sum_{\left(i,j\right)\in E\left(G\right)}\phi_{ij},
\]
which is the assertion. 
\end{proof}

\begin{cor}
\label{cor:b3}For complete graphs, $\delta=m-1$, hence $C\left(G\right)=1$
in \eqref{eq:B4}-\eqref{eq:bb3}. In this case, \prettyref{thm:c1}
reduces to \prettyref{thm:b1}. 
\end{cor}

\begin{rem}
The edge domination condition \eqref{eq:B3} is the natural hypothesis
needed to pass from the complete graph estimate to a graph local estimate.
It requires that each non-edge weight $\phi_{ij}$ be controlled by
the average of the neighboring edge weights at $i$ and at $j$, so
that the collection $(\phi_{ij})$ can be bounded globally in terms
of its values on $E(G)$. In particular, without some domination assumption
of this kind, \prettyref{exa:3-5} below shows that there is no finite
constant $C(G)$ for which a bound of the form \eqref{eq:B4} can
hold for all self-adjoint contractions. Thus the role of \eqref{eq:B3}
is to prevent large non-edge interactions from escaping the edge sum
on the right-hand side. 
\end{rem}

\subsection*{Necessity of Edge Domination}

The assumption \eqref{eq:B3} is essential. Without it, no finite
constant $C\left(G\right)$ can make an inequality of the form \eqref{eq:B4}
valid for all self-adjoint contractions.
\begin{example}
\label{exa:3-5}Let $m=3$ and let $G$ be the path graph with edge
set 
\[
E\left(G\right)=\left\{ \left(1,2\right),\left(2,3\right)\right\} ,
\]
so that $\delta=1$. Take $H=K=\mathbb{C}^{2}$ and define 
\[
x_{1}=y_{1}=\sigma_{z},\qquad x_{2}=y_{2}=0,\qquad x_{3}=y_{3}=\sigma_{x},
\]
where $\sigma_{x},\sigma_{y},\sigma_{z}$ are Pauli matrices. Thus
the vertex $2$ is an ``inactive'' site (the corresponding summand
is zero), so that all edge interactions vanish while the non-edge
$(1,3)$ carries a nontrivial contribution. Then for the edge $\left(1,2\right)$
we have 
\[
\left[x_{1},x_{2}\right]=0=\left[y_{1},y_{2}\right],\qquad\left\{ x_{1},x_{2}\right\} =0=\left\{ y_{1},y_{2}\right\} ,
\]
and similarly for the edge $\left(2,3\right)$ since $x_{2}=y_{2}=0$,
so 
\[
\phi_{12}=0,\qquad\phi_{23}=0.
\]

For the non-edge $\left(1,3\right)$, 
\[
\left\{ x_{1},x_{3}\right\} =0=\left\{ y_{1},y_{3}\right\} ,\qquad\left[x_{1},x_{3}\right]=2x_{1}x_{3}=\left[y_{1},y_{3}\right],
\]
giving 
\[
\phi_{13}=\frac{1}{2}\left(\left\Vert \left[x_{1},x_{3}\right]\right\Vert \left\Vert \left[y_{1},y_{3}\right]\right\Vert +\left\Vert \left\{ x_{1},x_{3}\right\} \right\Vert \left\Vert \left\{ y_{1},y_{3}\right\} \right\Vert \right)=\frac{1}{2}\left(2\cdot2+0\right)=2.
\]
Thus 
\[
\sum_{\left(i,j\right)\in E\left(G\right)}\phi_{ij}=0,\qquad\sum_{i<j}\phi_{ij}=2.
\]

In this example 
\[
B=\sum^{3}_{i=1}x_{i}\otimes y_{i}=\sigma_{z}\otimes\sigma_{z}+\sigma_{x}\otimes\sigma_{x}.
\]
As in the two-spin example below (\prettyref{exa:3-6}), $B$ has
eigenvalues $0,\pm2$, so 
\[
\left\Vert B\right\Vert =2,\qquad\left\Vert B\right\Vert ^{2}=4.
\]
The baseline estimate \prettyref{eq:bb5} gives 
\[
\left\Vert B\right\Vert ^{2}\leq m+\sum_{i<j}\phi_{ij}=3+2=5,
\]
while any sparse graph inequality of the form \prettyref{eq:B4} would
read 
\[
\left\Vert B\right\Vert ^{2}\leq m+C\left(G\right)\sum_{\left(i,j\right)\in E\left(G\right)}\phi_{ij}=3+0=3,
\]
which is contradicted by $\left\Vert B\right\Vert ^{2}=4$ for all
finite $C\left(G\right)$. 
\begin{conclusion*}
The assumption \prettyref{eq:B3} prevents uncontrolled non-edge interactions
from violating the global operator bound. 
\end{conclusion*}
\end{example}

\subsection*{Sharpness and Equality Patterns}

\subsubsection*{The two-term case}

Consider the simplest nontrivial instance of the complete graph inequality
(\prettyref{thm:b1}), corresponding to $m=2$.

Let $x_{1},x_{2}\in B\left(H\right)$ and $y_{1},y_{2}\in B\left(K\right)$
be self-adjoint unitaries satisfying the anticommutation relations
\[
\left\{ x_{1},x_{2}\right\} =0=\left\{ y_{1},y_{2}\right\} .
\]
Define 
\[
S:=x_{1}\otimes y_{1}+x_{2}\otimes y_{2}.
\]
Then 
\[
S^{2}=\left(x_{1}\otimes y_{1}\right)^{2}+\left(x_{2}\otimes y_{2}\right)^{2}+\left(x_{1}\otimes y_{1}\right)\left(x_{2}\otimes y_{2}\right)+\left(x_{2}\otimes y_{2}\right)\left(x_{1}\otimes y_{1}\right).
\]
Since $x^{2}_{i}=y^{2}_{i}=I$, this simplifies to 
\[
S^{2}=2I+x_{1}x_{2}\otimes y_{1}y_{2}+x_{2}x_{1}\otimes y_{2}y_{1}.
\]
Because $\left\{ x_{1},x_{2}\right\} =0$ and $\left\{ y_{1},y_{2}\right\} =0$,
we have $x_{2}x_{1}=-x_{1}x_{2}$ and $y_{2}y_{1}=-y_{1}y_{2}$. Hence,
\[
S^{2}=2\left(I+W\right),\qquad W:=x_{1}x_{2}\otimes y_{1}y_{2}.
\]

Since $x_{i},y_{i}$ are self-adjoint unitaries with $x^{2}_{i}=y^{2}_{i}=I$
and $x_{2}x_{1}=-x_{1}x_{2}$, 
\[
\left(x_{1}x_{2}\right)^{2}=x_{1}\left(x_{2}x_{1}\right)x_{2}=-I,
\]
and similarly $\left(y_{1}y_{2}\right)^{2}=-I$. It follows that 
\[
W^{2}=\left(x_{1}x_{2}\right)^{2}\otimes\left(y_{1}y_{2}\right)^{2}=\left(-I\right)\otimes\left(-I\right)=I,
\]
and 
\[
W^{*}=\left(x_{1}x_{2}\otimes y_{1}y_{2}\right)^{*}=x_{2}x_{1}\otimes y_{2}y_{1}=W.
\]
Thus $W$ is also a self-adjoint unitary with spectrum $\left\{ \pm1\right\} $,
and so $\left\Vert I+W\right\Vert =2$. 

Note that 
\[
\left\Vert S\right\Vert ^{2}=\left\Vert S^{2}\right\Vert =2\left\Vert I+W\right\Vert =2\times2=4,
\]
and hence 
\[
\left\Vert S\right\Vert =2.
\]

Recall that for $m=2$, \prettyref{thm:b1} gives 
\[
\left\Vert x_{1}\otimes y_{1}+x_{2}\otimes y_{2}\right\Vert ^{2}\leq2+\frac{1}{2}\left(\left\Vert \left[x_{1},x_{2}\right]\right\Vert \left\Vert \left[y_{1},y_{2}\right]\right\Vert +\left\Vert \left\{ x_{1},x_{2}\right\} \right\Vert \left\Vert \left\{ y_{1},y_{2}\right\} \right\Vert \right).
\]
Using that $\left\Vert \left\{ x_{1},x_{2}\right\} \right\Vert =\left\Vert \left\{ y_{1},y_{2}\right\} \right\Vert =0$,
and $\left\Vert \left[x_{1},x_{2}\right]\right\Vert =\left\Vert \left[y_{1},y_{2}\right]\right\Vert =2$,
the right-hand side is precisely $4$, matching the direct computation
above. Therefore, equality is achieved whenever both pairs $\left(x_{1},x_{2}\right)$
and $\left(y_{1},y_{2}\right)$ are anticommuting self-adjoint unitaries. 

This shows that the two-term case of the complete graph theorem is
exact and sharp.

\subsubsection*{Complete graphs and dense regimes}

As noted in \prettyref{cor:b3}, for the complete graph $G=K_{m}$,
the minimum degree is $\delta=m-1$, hence $C\left(G\right)=1$. The
sparse graph case (\prettyref{thm:c1}) reduces to the baseline inequality
\[
\left\Vert B\right\Vert ^{2}\leq m+\sum_{i<j}\phi_{ij}.
\]
If all pairs satisfy 
\[
\left\{ x_{i},x_{j}\right\} =0=\left\{ y_{i},y_{j}\right\} ,\qquad x^{2}_{i}=y^{2}_{i}=I,
\]
then 
\[
\phi_{ij}=\frac{1}{2}\left\Vert \left[x_{i},x_{j}\right]\right\Vert \left\Vert \left[y_{i},y_{j}\right]\right\Vert =2,
\]
and hence 
\[
\sum_{i<j}\phi_{ij}=2\binom{m}{2}=m\left(m-1\right),
\]
which gives 
\[
\left\Vert B\right\Vert ^{2}\leq m+m\left(m-1\right)=m^{2},
\]
i.e., $\left\Vert B\right\Vert \leq m$. This bound is attained in
the canonical Pauli examples (Examples \ref{exa:3-6} and \ref{exa:3-7}). 
\begin{example}[Two Spin Pairs]
\label{exa:3-6}Let $x_{1}=\sigma_{z}$, $x_{2}=\sigma_{x}$, $y_{1}=\sigma_{z}$,
$y_{2}=\sigma_{x}$, where $\sigma_{x},\sigma_{z}$ are the standard
Pauli matrices satisfying $\left\{ \sigma_{x},\sigma_{z}\right\} =0$
and $\sigma^{2}_{i}=I$. Then 
\[
B=x_{1}\otimes y_{1}+x_{2}\otimes y_{2}=\sigma_{z}\otimes\sigma_{z}+\sigma_{x}\otimes\sigma_{x}.
\]
Each tensor factor is self-adjoint and unitary, and the two summands
anticommute. A direct computation shows that $B$ has eigenvalues
$0,\pm2$. Hence $\left\Vert B\right\Vert =2=m$. 
\end{example}

\begin{example}[Three Pauli Generators]
\label{exa:3-7} Let $x_{i}=y_{i}\in\left\{ \sigma_{x},\sigma_{y},\sigma_{z}\right\} $.
Then 
\[
B=\sigma_{x}\otimes\sigma_{x}+\sigma_{y}\otimes\sigma_{y}+\sigma_{z}\otimes\sigma_{z}.
\]
Acting on $\mathbb{C}^{2}\otimes\mathbb{C}^{2}$, this operator is
the standard Heisenberg exchange coupling. It decomposes the space
into a three-dimensional triplet subspace with eigenvalue $1$, and
a one-dimensional singlet subspace with eigenvalue $-3$. Consequently,
$\left\Vert B\right\Vert =3=m$, giving equality in the complete graph
bound. 
\end{example}

\begin{example}[General Clifford family]
Let $\gamma_{1},\dots,\gamma_{m}$ be Hermitian Clifford generators
acting on some finite-dimensional Hilbert space, satisfying 
\[
\left\{ \gamma_{i},\gamma_{j}\right\} =0\;\left(i\neq j\right),\qquad\gamma^{2}_{i}=I.
\]
Set $x_{i}=y_{i}=\gamma_{i}$. Then 
\[
\left(\gamma_{i}\otimes\gamma_{i}\right)\left(\gamma_{j}\otimes\gamma_{j}\right)=\gamma_{i}\gamma_{j}\otimes\gamma_{i}\gamma_{j}=\left(-\gamma_{j}\gamma_{i}\right)\otimes\left(-\gamma_{j}\gamma_{i}\right)=\left(\gamma_{j}\otimes\gamma_{j}\right)\left(\gamma_{i}\otimes\gamma_{i}\right),
\]
so the family $\left\{ \gamma_{i}\otimes\gamma_{i}\right\} $ is commuting
and self-adjoint unitary. Thus these operators are simultaneously
diagonalizable, with joint eigenvalues 
\[
\left(s_{1},\dots,s_{m}\right)\in\Sigma\subseteq\left\{ \pm1\right\} ^{m},
\]
where $\Sigma$ is the joint spectrum.

Set 
\[
B=\sum^{m}_{i=1}\gamma_{i}\otimes\gamma_{i}.
\]
In a joint eigenbasis, $B$ is diagonal with eigenvalues 
\[
\lambda_{s}=\sum^{m}_{i=1}s_{i},\qquad\left(s_{1},\dots,s_{m}\right)\in\Sigma,
\]
so 
\[
\left\Vert B\right\Vert =\max_{s\in\Sigma}\left|\sum\nolimits^{m}_{i=1}s_{i}\right|\leq m.
\]
In particular, this matches the complete graph bound $\left\Vert B\right\Vert \leq m$.
\end{example}

\subsubsection*{A non-Clifford qubit example}

The next example shows that the complete graph bound applies beyond
the commuting/anticommuting Pauli-Clifford regime: for some pairs
$(i,j)$ neither $[x_{i},x_{j}]$ nor $\{x_{i},x_{j}\}$ vanishes,
so $\phi_{ij}$ uses both terms. 
\begin{example}
\label{exa:nonClifford} Let $H=K=\mathbb{C}^{2}$ and let $\sigma_{x},\sigma_{y},\sigma_{z}$
be the Pauli matrices. Fix an angle $\theta\in\left(0,\frac{\pi}{2}\right)$
and define 
\[
x_{1}=y_{1}=\sigma_{z},\qquad x_{2}=y_{2}=\sigma_{x},\qquad x_{3}=y_{3}=\cos\theta\,\sigma_{z}+\sin\theta\,\sigma_{x}.
\]
Each $x_{i}$ and $y_{i}$ is a self-adjoint unitary. Set 
\[
B=\sum^{3}_{i=1}x_{i}\otimes y_{i}.
\]

The pairs $(x_{1},x_{2})$ and $(y_{1},y_{2})$ are the usual Pauli
generators: 
\[
\{x_{1},x_{2}\}=0=\{y_{1},y_{2}\},\qquad[x_{1},x_{2}]=2i\sigma_{y}=[y_{1},y_{2}],
\]
so they anticommute. In contrast, for the pair $(x_{1},x_{3})$ one
computes 
\[
\{x_{1},x_{3}\}=\sigma_{z}\bigl(\cos\theta\,\sigma_{z}+\sin\theta\,\sigma_{x}\bigr)+\bigl(\cos\theta\,\sigma_{z}+\sin\theta\,\sigma_{x}\bigr)\sigma_{z}=2\cos\theta\,I,
\]
\[
[x_{1},x_{3}]=\sigma_{z}\bigl(\cos\theta\,\sigma_{z}+\sin\theta\,\sigma_{x}\bigr)-\bigl(\cos\theta\,\sigma_{z}+\sin\theta\,\sigma_{x}\bigr)\sigma_{z}=2i\sin\theta\,\sigma_{y},
\]
and similarly 
\[
\{y_{1},y_{3}\}=2\cos\theta\,I,\qquad[y_{1},y_{3}]=2i\sin\theta\,\sigma_{y}.
\]
For $\theta\in\left(0,\frac{\pi}{2}\right)$ both $\{x_{1},x_{3}\}$
and $[x_{1},x_{3}]$ are nonzero, so $x_{1}$ and $x_{3}$ neither
commute nor anticommute, and the same holds for $y_{1},y_{3}$.

The remaining pair $(x_{2},x_{3})$ also satisfies 
\[
\{x_{2},x_{3}\}=2\sin\theta\,I,\qquad[x_{2},x_{3}]=-2i\cos\theta\,\sigma_{y},
\]
with identical formulas for $(y_{2},y_{3})$. Thus, for each unordered
pair $i\neq j$ we have 
\[
\left\Vert [x_{i},x_{j}]\right\Vert =2\left|n_{i}\times n_{j}\right|,\qquad\left\Vert \{x_{i},x_{j}\}\right\Vert =2\left|n_{i}\cdot n_{j}\right|,
\]
where $n_{i}\in\mathbb{R}^{3}$ is the Bloch direction of $x_{i}$,
and likewise for the $y_{i}$. In particular, for all pairs $(i,j)$
one checks $\left|n_{i}\times n_{j}\right|^{2}+\left|n_{i}\cdot n_{j}\right|^{2}=1$,
so 
\[
\phi_{ij}=\frac{1}{2}\Bigl(\left\Vert [x_{i},x_{j}]\right\Vert \left\Vert [y_{i},y_{j}]\right\Vert +\left\Vert \{x_{i},x_{j}\}\right\Vert \left\Vert \{y_{i},y_{j}\}\right\Vert \Bigr)=2
\]
for every $i\neq j$.

Applying the complete graph inequality \prettyref{thm:b1} with $m=3$
and $H=K=\mathbb{C}^{2}$ gives 
\[
\left\Vert B\right\Vert ^{2}\leq m+\sum_{1\leq i<j\leq3}\phi_{ij}=3+3\cdot2=9,\qquad\text{so}\qquad\left\Vert B\right\Vert \leq3.
\]
In this example the families $\{x_{i}\}$, $\{y_{i}\}$ are no longer
Clifford: for the pairs $(1,3)$ and $(2,3)$ neither the commutator
nor the anticommutator vanishes, and the bound uses both contributions
through the coefficients $\phi_{ij}$. 
\end{example}

\section{Weighted Inequalities}\label{sec:3}

This section extends the unweighted tensor‐sum bounds of \prettyref{sec:2}
to the weighted case. Each weight $c_{i}\in\mathbb{R}$ corresponds
to the strength of a local measurement setting, allowing the inequalities
to quantify attainable bipartite correlations and to infer structural
constraints on non-commutativity.

\subsection{Weighted Complete Graph Bound}

Let $x_{i}\in B\left(H\right)$, $y_{i}\in B\left(K\right)$ be self-adjoint
contractions and $c_{i}\in\mathbb{R}$. Define the weighted sum 
\[
B_{c}=\sum^{m}_{i=1}c_{i}x_{i}\otimes y_{i}.
\]
For each unordered pair $i\neq j$, set 
\[
\phi_{ij}=\frac{1}{2}\left(\left\Vert \left[x_{i},x_{j}\right]\right\Vert \left\Vert \left[y_{i},y_{j}\right]\right\Vert +\left\Vert \left\{ x_{i},x_{j}\right\} \right\Vert \left\Vert \left\{ y_{i},y_{j}\right\} \right\Vert \right)\geq0.
\]
 
\begin{thm}[weighted complete graph inequality]
\label{thm:C1} 
\begin{equation}
\left\Vert B_{c}\right\Vert ^{2}\leq\sum^{m}_{i=1}c^{2}_{i}+\sum_{i<j}\left|c_{i}c_{j}\right|\phi_{ij}.\label{eq:c1}
\end{equation}
 
\end{thm}

\begin{proof}
Let $u_{i}=x_{i}\otimes y_{i}$, then 
\[
B^{2}_{c}=\sum_{i}c^{2}_{i}u^{2}_{i}+\sum_{i<j}c_{i}c_{j}\left(u_{i}u_{j}+u_{j}u_{i}\right).
\]
Since each $x_{i},y_{i}$ is a self-adjoint contraction, $u^{2}_{i}=x^{2}_{i}\otimes y^{2}_{i}\leq I$.
Hence 
\[
\sum_{i}c^{2}_{i}u^{2}_{i}\leq\left(\sum_{i}c^{2}_{i}\right)I.
\]
For $i<j$, 
\[
u_{i}u_{j}+u_{j}u_{i}=\frac{1}{2}\left(\left\{ x_{i},x_{j}\right\} \otimes\left\{ y_{i},y_{j}\right\} +\left[x_{i},x_{j}\right]\otimes\left[y_{i},y_{j}\right]\right)=:T_{ij},
\]
and $T^{*}_{ij}=T_{ij}$. It follows that 
\[
\left\Vert T_{ij}\right\Vert \leq\frac{1}{2}\left(\left\Vert \left[x_{i},x_{j}\right]\right\Vert \left\Vert \left[y_{i},y_{j}\right]\right\Vert +\left\Vert \left\{ x_{i},x_{j}\right\} \right\Vert \left\Vert \left\{ y_{i},y_{j}\right\} \right\Vert \right)=\phi_{ij}.
\]
Since $B_{c}$ is self-adjoint, $\left\Vert B_{c}\right\Vert ^{2}=\left\Vert B^{2}_{c}\right\Vert $.
Applying the triangle inequality, we get 
\[
\left\Vert B_{c}\right\Vert ^{2}=\left\Vert B^{2}_{c}\right\Vert \leq\left\Vert \sum c^{2}_{i}u^{2}_{i}\right\Vert +\left\Vert \sum c_{i}c_{j}T_{ij}\right\Vert \leq\sum_{i}c^{2}_{i}+\sum_{i<j}\left|c_{i}c_{j}\right|\phi_{ij},
\]
which is \eqref{eq:c1}. 
\end{proof}

\begin{cor}
For any bipartite state $\rho$ on $H\otimes K$, 
\[
\left|\mathrm{tr}\left(\rho B_{c}\right)\right|\leq\left\Vert B_{c}\right\Vert \leq\sqrt{\sum_{i}c^{2}_{i}+\sum_{i<j}\left|c_{i}c_{j}\right|\phi_{ij}}.
\]
\end{cor}

\begin{proof}
As above, $\|\rho\|_{1}=1$, so $\left|\mathrm{tr}\left(\rho B_{c}\right)\right|\leq\|\rho\|_{1}\,\|B_{c}\|=\|B_{c}\|$,
and the upper bound is \prettyref{thm:C1}. 
\end{proof}

\subsection{Weighted Sparse Graph Bound}

Let $G=\left(V,E\right)$ be a simple undirected graph on $V=\left\{ 1,\dots,m\right\} $
with minimum degree $\delta\geq1$. The edges $E$ indicate allowed
pairwise interactions among measurement settings. 
\begin{thm}[weighted sparse graph Bell inequality]
\label{thm:C3}Suppose, for every non-edge pair $\left(i,j\right)$,
the following weighted edge-domination condition holds: 
\begin{equation}
\left|c_{i}c_{j}\right|\phi_{ij}\leq\frac{1}{\deg\left(i\right)}\sum_{k\in N\left(i\right)}\left|c_{i}c_{k}\right|\phi_{ik}+\frac{1}{\deg\left(j\right)}\sum_{k\in N\left(j\right)}\left|c_{j}c_{k}\right|\phi_{jk}.\label{eq:c2}
\end{equation}
Then, 
\[
\left\Vert B_{c}\right\Vert ^{2}\leq\sum_{i}c^{2}_{i}+C\left(G\right)\sum_{\left(i,j\right)\in E\left(G\right)}\left|c_{i}c_{j}\right|\phi_{ij},
\]
where 
\[
C\left(G\right)=\frac{2\left(m-1\right)}{\delta}-1.
\]
 
\end{thm}

\begin{proof}
From \prettyref{thm:C1},
\begin{equation}
B^{2}_{c}\leq\left(\sum_{i}c^{2}_{i}+\sum_{i<j}\left|c_{i}c_{j}\right|\phi_{ij}\right)I.\label{eq:c3}
\end{equation}
Split the pair sum into edges and non-edges:
\begin{equation}
\sum_{i<j}\left|c_{i}c_{j}\right|\phi_{ij}=\sum_{i<j,\:\left(i,j\right)\in E}\left|c_{i}c_{j}\right|\phi_{ij}+\underset{=S_{NE}}{\underbrace{\sum_{i<j,\:\left(i,j\right)\notin E}\left|c_{i}c_{j}\right|\phi_{ij}}}.\label{eq:c4}
\end{equation}

Apply \eqref{eq:c2} to each non-edge and sum over all such pairs,
we get 
\[
S_{NE}\leq\sum_{i<j,\:\left(i,j\right)\notin E}\left(\frac{1}{\deg\left(i\right)}\sum_{k\in N\left(i\right)}\left|c_{i}c_{k}\right|\phi_{ik}+\frac{1}{\deg\left(j\right)}\sum_{k\in N\left(j\right)}\left|c_{j}c_{k}\right|\phi_{jk}\right).
\]
A similar argument as in the proof of \prettyref{thm:c1} gives 
\[
S_{NE}\leq\sum^{m}_{i=1}\frac{m-1-\deg\left(i\right)}{\deg\left(i\right)}\sum_{k\in N\left(i\right)}\left|c_{i}c_{k}\right|\phi_{ik}.
\]

Since $\deg\left(i\right)\geq\delta$, 
\[
\frac{m-1-\deg\left(i\right)}{\deg\left(i\right)}\leq\frac{m-1}{\delta}-1.
\]
Moreover, 
\[
\sum_{i}\sum_{k\in N\left(i\right)}\left|c_{i}c_{k}\right|\phi_{ik}=2\sum_{\left(i,j\right)\in E}\left|c_{i}c_{j}\right|\phi_{ij}.
\]
Hence 
\begin{equation}
S_{NE}\leq2\left(\frac{m-1}{\delta}-1\right)\sum_{\left(i,j\right)\in E}\left|c_{i}c_{j}\right|\phi_{ij}.\label{eq:c5}
\end{equation}

Combining \eqref{eq:c3}--\eqref{eq:c5} gives 
\[
B^{2}_{c}\leq\left(\sum_{i}c^{2}_{i}+\left(\frac{2\left(m-1\right)}{\delta}-1\right)\sum_{\left(i,j\right)\in E\left(G\right)}\left|c_{i}c_{j}\right|\phi_{ij}\right)I,
\]
and taking norms yields the claim. 
\end{proof}

\begin{example}
For $G$ a star $\left(\delta=1\right)$, 
\[
\left\Vert B_{c}\right\Vert ^{2}\leq\sum_{i}c^{2}_{i}+\left(2m-3\right)\sum^{m}_{j=2}\left|c_{1}c_{j}\right|\phi_{1j}.
\]
For $G$ a chain $\left(\delta=1\right)$, 
\[
\left\Vert B_{c}\right\Vert ^{2}\leq\sum_{i}c^{2}_{i}+\left(2m-3\right)\sum^{m-1}_{i=1}\left|c_{i}c_{i+1}\right|\phi_{i,i+1}.
\]
In particular, if $\phi_{i,i+1}\ll1$ for all edges, then $\left\Vert B_{c}\right\Vert \approx\sqrt{\sum_{i}c^{2}_{i}}$
even for long chains. 
\end{example}

\begin{cor}
\label{cor:C5}Let $\beta=\sup_{\rho}\left|\text{tr}\left(\rho B_{c}\right)\right|$
be the observed Bell value (supremum over all bipartite states $\rho$).
Then for the complete graph case,
\begin{equation}
\sum_{i<j}\left|c_{i}c_{j}\right|\phi_{ij}\geq\beta^{2}-\sum_{i}c^{2}_{i},\label{eq:c7}
\end{equation}
and under the assumption \eqref{eq:c2}, 
\begin{equation}
\sum_{\left(i,j\right)\in E\left(G\right)}\left|c_{i}c_{j}\right|\phi_{ij}\geq\frac{1}{C\left(G\right)}\left(\beta^{2}-\sum_{i}c^{2}_{i}\right),\qquad C\left(G\right):=\frac{2\left(m-1\right)}{\delta}-1.\label{eq:c8}
\end{equation}
\end{cor}

\begin{proof}
Hölder's inequality gives $\left|\text{tr}\left(\rho B_{c}\right)\right|\leq\left\Vert B_{c}\right\Vert \left\Vert \rho\right\Vert _{1}=\left\Vert B_{c}\right\Vert $,
hence $\beta\leq\left\Vert B_{c}\right\Vert $. For complete graphs,
\prettyref{thm:C1} gives 
\[
\sum_{\left(i,j\right)\in E\left(G\right)}\left|c_{i}c_{j}\right|\phi_{ij}\geq\left\Vert B_{c}\right\Vert ^{2}-\sum_{i}c^{2}_{i}\geq\beta^{2}-\sum_{i}c^{2}_{i},
\]
as claimed in \eqref{eq:c7}. 

For sparse graphs, \prettyref{thm:C3} gives 
\[
C\left(G\right)\sum_{\left(i,j\right)\in E\left(G\right)}\left|c_{i}c_{j}\right|\phi_{ij}\geq\left\Vert B_{c}\right\Vert ^{2}-\sum_{i}c^{2}_{i}.
\]
Using $\beta\leq\left\Vert B_{c}\right\Vert $ again, we get 
\[
\sum_{\left(i,j\right)\in E\left(G\right)}\left|c_{i}c_{j}\right|\phi_{ij}\geq\frac{1}{C\left(G\right)}\left(\beta^{2}-\sum_{i}c^{2}_{i}\right),
\]
which is \eqref{eq:c8}.
\end{proof}

\begin{rem}
If $\beta^{2}\leq\sum_{i}c^{2}_{i}$, then the right-hand side is
non-positive. In that case the inequalities are trivial since the
left-hand sides are nonnegative by definition. One may state the bounds
equivalently as
\begin{align*}
\sum_{i<j}\left|c_{i}c_{j}\right|\phi_{ij} & \geq\max\left\{ 0,\beta^{2}-\sum c^{2}_{i}\right\} ,\\
\sum_{\left(i,j\right)\in E\left(G\right)}\left|c_{i}c_{j}\right|\phi_{ij} & \geq\max\left(0,\frac{1}{C\left(G\right)}\left(\beta^{2}-\sum c^{2}_{i}\right)\right).
\end{align*}
\end{rem}

\begin{cor}
\label{cor:C7}Fix a threshold $t>0$ and assume $\left|c_{i}\right|\leq c_{max}$
for all $i$. Let 
\[
\beta:=\sup_{\rho}\left|\mathrm{tr}\left(\rho B_{c}\right)\right|,\qquad S:=\max\left\{ 0,\beta^{2}-\sum_{i}c^{2}_{i}\right\} .
\]
Set $M:=4c^{2}_{max}$.
\begin{enumerate}
\item Complete graph case. Define 
\[
N_{t}:=\left|\left\{ \left(i,j\right):1\leq i<j\leq m,\:\left|c_{i}c_{j}\right|\phi_{ij}\geq t\right\} \right|.
\]
If $0<t<M$, then 
\begin{equation}
N_{t}\geq\max\left\{ 0,\frac{S-t\binom{m}{2}}{M-t}\right\} .\label{eq:C9}
\end{equation}
\item Sparse graph case. Let $G=\left(V,E\right)$ have minimum degree $\delta\geq1$
and $C\left(G\right)=\frac{2\left(m-1\right)}{\delta}-1$. Define
\[
N^{E}_{t}:=\left|\left\{ \left(i,j\right)\in E:\left|c_{i}c_{j}\right|\phi_{ij}\geq t\right\} \right|.
\]
If $0<t<M$, then 
\begin{equation}
N^{E}_{t}\geq\max\left\{ 0,\frac{\frac{S}{C\left(G\right)}-t\left|E\right|}{M-t}\right\} .\label{eq:C10}
\end{equation}
\end{enumerate}
\end{cor}

\begin{proof}
For self-adjoint contractions one has $\left\Vert \left[x_{i},x_{j}\right]\right\Vert \leq2$
and $\left\Vert \left\{ x_{i},x_{j}\right\} \right\Vert \leq2$, and
likewise for the $y$'s, hence $\phi_{ij}\leq4$. With $\left|c_{i}\right|\leq c_{max}$
this gives $\left|c_{i}c_{j}\right|\phi_{ij}\leq4c^{2}_{max}=M$.

Complete graph case. By \prettyref{cor:C5}, 
\[
\sum_{i<j}\left|c_{i}c_{j}\right|\phi_{ij}\geq S.
\]
Let $A=\left\{ \left(i,j\right):\left|c_{i}c_{j}\right|\phi_{ij}\geq t\right\} $,
so $\left|A\right|=N_{t}$. Then 
\[
\sum_{i<j}\left|c_{i}c_{j}\right|\phi_{ij}=\sum_{\left(i,j\right)\in A}\left|c_{i}c_{j}\right|\phi_{ij}+\sum_{\left(i,j\right)\notin A}\left|c_{i}c_{j}\right|\phi_{ij}\leq MN_{t}+t\left(\binom{m}{2}-N_{t}\right).
\]
Combining with the lower bound $S$ yields 
\[
S\leq t\binom{m}{2}+\left(M-t\right)N_{t},
\]
which gives \prettyref{eq:C9}.

Sparse graph case. By \prettyref{cor:C5} one has 
\[
\sum_{\left(i,j\right)\in E}\left|c_{i}c_{j}\right|\phi_{ij}\geq\frac{S}{C\left(G\right)}.
\]
Repeat the same argument but summing over edges of $G$ to obtain
\[
\frac{S}{C\left(G\right)}\leq t\left|E\right|+\left(M-t\right)N^{E}_{t},
\]
which gives \prettyref{eq:C10}. 
\end{proof}

\bibliographystyle{amsalpha}
\bibliography{ref}

\end{document}